\documentclass[11pt]{article}

\usepackage{amsmath, amssymb,amsfonts,xspace,graphicx,relsize,bm,mathtools,xcolor}
\usepackage{mathrsfs}
\usepackage[pagebackref]{hyperref}
\usepackage{enumerate}
\usepackage{cleveref}

\def\01{\{0,1\}}

\newcommand{\eps}{\varepsilon}
\newcommand{\ket}[1]{|#1\rangle}
\newcommand{\bra}[1]{\langle#1|}
\newcommand{\ketbra}[2]{|#1\rangle\langle#2|}

\newcommand{\braket}[2]{\langle#1|#2\rangle}
\newcommand{\Tr}{\mbox{\rm Tr}}

\newcommand{\Cc}{{\pazocal C}} %

\usepackage{calrsfs}
\DeclareMathAlphabet{\pazocal}{OMS}{zplm}{m}{n}

\newcommand{\Exp}{\mathbb{E}}

\newcommand{\Rc}{\ensuremath{\mathsf{R}}}

\newcommand{\Qsq}{\ensuremath{\mathsf{QSQ}}}
\newcommand{\VC}{\ensuremath{\mathsf{VC}}}
\newcommand{\Sq}{\ensuremath{\mathsf{SQ}}}
\newcommand{\Qc}{\ensuremath{\mathsf{Q}}}

\newcommand{\Hi}{\ensuremath{\pazocal{H}}}
\newcommand{\Ch}{\ensuremath{\mathcal{H}}}

\newcommand{\A}{\ensuremath{\mathcal{A}}}

\newcommand{\R}{\ensuremath{\mathbb{R}}}
\newcommand{\C}{\ensuremath{\mathbb{C}}}
\newcommand{\id}{\ensuremath{\mathbb{I}}}
	\newcommand{\PARITY}{\mathsf{PARITY}}
	\newcommand{\size}{\mathsf{size}}
	
	\newcommand{\disc}{\mathsf{disc}}
	\newcommand{\uniform}{\mathcal{U}}

\DeclareMathOperator{\poly}{poly}

\newcommand{\Fe}{\ensuremath{\pazocal{F}}}

\newtheorem{theorem}{Theorem}[section]
\newtheorem{definition}[theorem]{Definition}

\newtheorem{lemma}[theorem]{Lemma}

\newtheorem{corollary}[theorem]{Corollary}

\newcommand{\pmset}[1]{\{-1,1\}^{#1}} %

\def\01{\{0,1\}}
\DeclareMathOperator{\Inf}{Inf}

\newcommand{\Qstat}{\mathsf{Qstat}}
\newcommand{\Stat}{\mathsf{Stat}}
\newcommand{\WSQDIM}{\mathsf{WeakSQDIM}}
\newcommand{\Prdim}{\mathsf{PRDIM}}
\newcommand{\supp}{\mathsf{supp}}
\newcommand{\citewithname}[2]{#1~#2}

\DeclareMathOperator{\sign}{sign}

\usepackage[noline, boxed]{algorithm2e}

\usepackage{kpfonts}

\usepackage[margin=1in]{geometry}
\hypersetup{
	colorlinks,
	linkcolor={blue!100!black},
	citecolor={blue!100!black},
}

\newenvironment{proof}
{\noindent {\bf Proof. }}
{{\hfill $\Box$}\\
	\smallskip}
\begin{document}
	
	\title{Quantum statistical query learning}
	\author{Srinivasan Arunachalam\thanks{IBM Research. {\tt Srinivasan.Arunachalam@ibm.com}}
\and
Alex B. Grilo\thanks{QuSoft and CWI, Amsterdam. {\tt alexg@cwi.com}
}
	\and
	Henry Yuen\thanks{University of Toronto. {\tt hyuen@cs.toronto.edu}}
	}
	\date{}
	\maketitle
	\begin{abstract}

 We propose a learning model called the \emph{quantum statistical learning ($\Qsq$)} model, which extends the $\Sq$ learning model introduced by Kearns~\cite{kearns:statistical} to the quantum setting. Our model can be also seen as a restriction of the quantum PAC learning model: here, the learner does not have direct access to quantum examples, but can only obtain estimates of measurement statistics on them. Theoretically, this model provides {a simple yet expressive setting} to explore the power of quantum examples in machine learning. From a practical perspective, since {simpler} operations are required, learning algorithms in the $\Qsq$ model are {more} feasible {for} %
 implementation on near-term quantum~devices. 

We prove a number of results about the $\Qsq$ learning model. We first show that parity functions, $O(\log n)$-juntas and polynomial-sized DNF formulas are efficiently learnable in the $\Qsq$ model, in contrast to the classical setting where these problems are provably hard. This implies that many of the advantages of quantum PAC learning can be realized even in the more restricted quantum $\Sq$ learning model.

{It is well-known that \emph{weak statistical query dimension}, denoted by $\WSQDIM(\Cc)$, characterizes the complexity of learning a concept class $\Cc$ in the classical $\Sq$ model. We show that $\log(\WSQDIM(\Cc))$ is a lower bound on the complexity of $\Qsq$ learning, and furthermore it is tight for certain concept classes $\Cc$.  Additionally, we show that this quantity provides strong lower bounds for the small-bias quantum communication model under product distributions.}

Finally, we introduce the notion of \emph{private} quantum  PAC learning, in which a quantum PAC learner is required to be \emph{differentially private}. We show that learnability in the $\Qsq$ model implies learnability in the quantum private PAC model. Additionally, we show that in the private PAC learning setting, the classical and quantum sample complexities are equal, up to constant factors.

 	\end{abstract}

	\section{Introduction}
The prospect of using quantum computers to perform machine learning has received much attention lately, given their potential to offer significant speedups for solving certain
problems of practical relevance. There has been a flurry of proposed quantum algorithms for performing computations that are ubiquitous in machine learning, ranging from convex optimization, matrix completion,  clustering, support vector machines~\cite{kerenidis2016quantum,brandao2017quantum,lloyd:clustering,rebentrost:SVM}. Due to the assumptions required by these quantum algorithms, the evidence for a quantum computational advantage in performing machine learning tasks is murky at best~\cite{tang:recommendation,chia:dequantizeall}.     
It is therefore an active area of research to obtain evidence (even conditional) for quantum advantage in machine learning.

Quantum learning theory has provided a theoretical framework to study the capabilities and limitations of quantum machine learning.  Here, the focus is not only on the computational complexity of learning algorithms, but also on \emph{information-theoretic} measures such as sample and query complexity. One of the first classical learning models that were generalized to the quantum setting was Valiant's Probably Approximately Correct (PAC) model of learning~\cite{valiant:paclearning}. In the classical PAC model of learning, the goal is to learn a collection of Boolean functions, which is often referred to as a \emph{concept class} $\Cc\subseteq \{c:\01^n\rightarrow \01\}$. The elements of a concept class are called \emph{concepts}. In the PAC model of learning, there is an unknown distribution $D:\01^n\rightarrow [0,1]$ and a learner is given \emph{labelled examples} $\{x_i,c^*(x_i)\}_i$ where $x_i$ is drawn from the distribution $D$ and $c^*\in \Cc$ is the unknown \emph{target concept}. The goal of a learner is the following: for every unknown $D$ and~$c^*$, use labelled examples to produce a hypothesis $h$ that satisfies $\Pr_{x\sim D}[h(x)=c^*(x)]\geq 2/3$. 
The \emph{quantum PAC model}, introduced by \citewithname{Bshouty and Jackson}{\cite{bshoutyandjackson:DNF}}, considers the extension of Valiant's PAC model where the learning algorithm is not given  labelled examples $\{x_i,c(x_i)\}_i$,
but instead is given copies of a \emph{quantum example} 
\[
    \ket{\psi_{c^*}} = \sum_{x \in \{0,1\}^n} \sqrt{D(x)} \, \ket{x,c^*(x)}.
\]
which is a \emph{superposition} of labeled examples. Observe that simply measuring $\ket{\psi_{c^*}}$ in the computational basis gives a classical labelled example. Quantum examples are well-motivated in quantum computing: they arise naturally in quantum query algorithms, and also have interesting complexity-theoretic applications~\cite{AharonovT07}.

In the distribution-independent PAC learning model, \citewithname{Arunachalam and de Wolf}{\cite{arunachalam:optimalpaclearning}} showed that the sample complexity of quantum and classical PAC learning is the same. However, in the uniform distribution learning model (i.e., when we fix $D$ to be the uniform distribution), quantum examples %
have been shown to be very powerful. In particular, given uniform quantum examples $\frac{1}{\sqrt{2^n}}\sum_x\ket{x,c^*(x)}$ a quantum learner {can efficiently} sample from the Fourier distribution  $\{\widehat{c^*}(S)^2\}_S$, {a tool that} has been used to provide even exponential advantage over the known classical algorithms~\cite{bshoutyandjackson:DNF,atici&servedio:qlearning,grilo:LWEeasy,arunachalam:qexactlearning}.

In this paper, we further investigate the power of quantum examples in learning, by defining a restricted quantum learning model and studying its capabilities and limitations. We call it the \emph{quantum statistical query ($\Qsq$)} model , which extends the well-studied (classical) \emph{statistical query} ($\Sq$) learning model introduced by Kearns~\cite{kearns:statistical}. In $\Sq$ learning, the learner constructs a hypothesis not by examining a sequence of labelled examples, but instead by adaptively querying an oracle to obtain \emph{estimates} of statistical properties of the labelled examples. Though this model is weaker than PAC learning, it is rich enough to capture many known learning algorithms~\cite{blum:intro2,Feldman2008,feldman:completechar}.
\paragraph{Quantum statistical query model.} In the $\Qsq$ model, the learner -- which is still a \emph{classical} randomized algorithm -- can query an oracle to obtain statistics of quantum examples to compute a hypothesis. {
Roughly speaking, these statistics correspond to the average value obtained if a quantum computer would repeatedly measure copies of quantum examples using a specified measurement $M$.
More concretely, in quantum computing measurements are defined by Hermitian matrices called observables,  
and the statistics obtained by the learner consist of an estimate of the expectation value $ \bra{\psi_{c^*}} M \ket{\psi_{c^*}}$ for a chosen observable~$M$.} When $M$ is diagonal, this reduces to the case of making classical $\Sq$ queries, and the power of $\Qsq$ appears when $M$ corresponds to measurements of $\ket{\psi_{c^*}}$ in a non-classical basis.%

We motivate the study of this model in several ways. Some concept classes appear to be learned more efficiently in the quantum PAC setting (at least in the \emph{distribution-dependent} setting); a natural question is whether these efficiency gains come from the ability of the quantum learning algorithm to \emph{directly} manipulate coherent superpositions of labeled data (i.e., quantum examples), or does the weaker quantum statistical query access suffice? Are there classical-quantum learning separations even in this weak statistical query model, where the learner can only access the data through measurement statistics of quantum examples? 

Another motivation comes from the consideration that $\Qsq$ learners are more practically feasible than general quantum PAC learners.  A general quantum PAC learning algorithm could perform complex entangling unitaries and measurements on many quantum examples simultaneously in order to extract joint statistics. However, this seems far beyond the capabilities of noisy, near-term quantum computers.
In the $\Qsq$ learning model, the learner can only obtain statistics about individual quantum examples. In a practical implementation of these quantum learning algorithms, this would only require measuring a single quantum example at a time. One could imagine a scenario where classical learning algorithms can query a cloud-based quantum computer to solve a learning task; the $\Qsq$ model would lend itself naturally to this situation.

\paragraph{Our contributions.} The first contribution of our paper is providing a definition for the quantum statistical query model.  We then prove a number of results regarding this model.

\begin{enumerate}
    \item We show that a query-efficient $\Qsq$ learner for a concept class $\Cc$ under a distribution $D$ implies a sample-efficient quantum PAC learner for $\Cc$ under the same distribution, and furthermore implies a sample-efficient \emph{noisy} quantum PAC learning of $\Cc$. This is exactly analogous to how classical $\Sq$ learning is a restriction of noisy PAC learning, which is {itself} a restriction of standard PAC learning. 
    
    \item We present three learning problems that can be solved efficiently in the $\Qsq$ model, but not in the classical $\Sq$ model. In particular, we show that it is possible to learn parity functions, juntas, and DNF formulas under the uniform distribution in polynomial time in the $\Qsq$ model; in contrast, the same problems are provably hard in the classical $\Sq$ model. {Notice that for juntas and DNFs, no efficient classical learning algorithm is currently known even in the setting where the learner is given the classical samples.}
    
    \item We show that 
while the  statistical query dimension characterizes the query complexity in the \Sq{} model, its logarithm is a tight lower-bound to the query complexity of $\Qsq$ learning. We also show a connection between  one-way communication complexity (under product distributions) with weak statistical query dimension. {In particular, this connection allows us to prove non-trivial lower bounds on the communication complexity (under product distributions) even for {\em inverse exponential bias} in computing the  function value.}

\item 	Our final contribution in this paper is to define the notion of \emph{privacy} in the quantum PAC learning model. We then lift the {fundamental} connection between  classical statistical query learning and private PAC learning to the quantum setting and show that learnability in the quantum $\Sq$ model implies  private quantum PAC learnability. Finally, we provide a  combinatorial characterization of the sample complexity of private quantum PAC learning and using this characterization we show that the sample complexities of private classical and quantum PAC learning are equal, up to constant factors.
\end{enumerate}

 \paragraph{Acknowledgements.} We thank Sasho Nikolov and Tanay Mehta for useful discussions. SA~did part of this work~at MIT and was supported by the MIT-IBM Watson AI Lab under the project \emph{Machine Learning in
Hilbert~space}. HY~was supported by NSERC Discovery Grant 2019-06636. Part of this work~was done while AG~and HY~were visiting the Simons Institute for the Theory of~Computing.

\paragraph{Organization.} 
In Section~\ref{sec:defnqsqmodel} we introduce the quantum statistical query model. In Section~\ref{sec:introlearneff} we present three concept classes that are efficiently learnable in the quantum statistical query model. We follow {by showing a lower bound} to the query complexity in the $\Qsq$ model and its relations to communication complexity in Section~\ref{sec:introlowerbounds}. Finally, in Section~\ref{sec:introdp} we present connections between the $\Qsq$ model and {quantum differential~privacy}.

	\section{Preliminaries}
	\label{sec:introprelim}
	We let $[n]=\{1,\ldots,n\}$. For $s\in \01^n$, define $\supp(s)=\{i\in [n]:s_i=1\}$. For $S\subseteq [n]$, denote $S^c=[n]\backslash S$ be the complement of $S$.
	
\paragraph{Quantum computing.}
We briefly review the basic concepts in quantum computing. We define $\ket{0}=\left(\begin{array}{c}1 \\ 0 \end{array}\right)$ and $\ket{1}=\left(\begin{array}{c} 0 \\ 1 \end{array}\right)$ as the canonical basis for $\mathbb{C}^2$. A qubit $\ket{\psi}$ is a unit vector in $\mathbb{C}^2$, i.e., $\alpha\ket{0}+\beta\ket{1}$ for $\alpha,\beta\in \mathbb{C}$ that satisfy $|\alpha|^2+|\beta|^2=1$. Multi-qubit quantum states are obtained   by taking tensor products of single-qubit states: an arbitrary $n$-qubit quantum state $\ket{\psi}\in \mathbb{C}^{2^n}$ is a unit vector in $\mathbb{C}^{2^n}$ and can be expressed as $\ket{\psi}=\sum_{x\in \01^n}\alpha_x \ket{x}$ where $\alpha_x\in \mathbb{C}$ and $\sum_x |\alpha_x|^2=1$. {We denote by $\bra{\psi}$ as the conjugate transpose of the quantum state~$\ket{\psi}$}. On a quantum computer one is allowed to arbitrary quantum gates (or operations) that correspond to unitary matrices. One gate we use often is the Hadamard gate, defined as $
\mathsf{H} = 
\frac{1}{\sqrt{2}}\begin{pmatrix}
1 & 1  \\
1 & -1
\end{pmatrix}.
$
 An \emph{observable} $M$ is a Hermitian matrix, which encodes a measurement in quantum mechanics. The average measurement outcome of a state $\ket{\psi}$ using the observable $O$ is given by the expectation value $\bra{\psi}O\ket{\psi}$.

\paragraph{Fourier analysis.}	We now introduce the basics of Fourier analysis on the Boolean cube. For $S\in~\01^n$, we define the \emph{character function} $\chi_S:\01^n\rightarrow \pmset{}$ as $\chi_S(x)=(-1)^{S\cdot x}$ where $S\cdot x=\sum_i S_i\cdot x_i \pmod 2$. For $f:\01^n\rightarrow \pmset{}$, the Fourier coefficients of $f$ are
	$$
	\widehat{f}(S)=\Exp_{x\in \01^n}[ f(x)\cdot \chi_S(x)] \quad \text{ for every } S\in \01^n,
	$$
	where the expectation is taken with respect to the uniform distribution over $\01^n$. Every function  $f:\01^n\rightarrow \R$ can be written uniquely as 
	$
	f(x)=\sum_{S\in \01^n}\widehat{f}(S)\chi_S(x)$. 
	Parseval's identity states that  $\sum_S \widehat{f}(S)^2=\Exp[f(x)^2]=~1$. Hence, $\{\widehat{f}(S)^2\}_{S\in \01^n}$ forms a probability distribution. For every $i\in [n]$, we define the $i$th influence as
	$$
	\Inf_i(f)=\sum_{\substack{S\in \01^n: \\ S_i=1}}\widehat{f}(S)^2.
	$$

\subsection{PAC learning}	
\label{sec:defnoflearning}
Valiant \cite{valiant:paclearning} introduced the Probably Approximately Correct (PAC) model of learning, which gives a formalization of what ``learning a function'' means. In this learning model, a \emph{concept class} $\Cc$ is a collection of Boolean functions $\Cc\subseteq \{c:\01^n\rightarrow \01\}$. The functions inside $\Cc$ are referred to as \emph{concepts}. Let $D:\01^n\rightarrow [0,1]$ be an \emph{unknown distribution} over the Boolean cube. In the PAC model, a learner $\A$ is given many \emph{labelled examples} $(x,c(x))$ where $x$ is drawn from the distribution $D$ and $c\in \Cc$ is the \emph{unknown} target concept. The goal of an $(\varepsilon,\delta)$-learner is the following: with probability at least $1-\delta$ (probability taken according to internal randomness of $\A$ and $D$), output  a hypothesis $h:\01^n\rightarrow \01$ that satisfies $\Pr_{x\sim D}[h(x)=c(x)]\geq 1-\varepsilon$. The $(\varepsilon,\delta)$-sample complexity of a learning algorithm $\A$ is the maximal number of labelled examples used, maximized over all $c\in \Cc$ and distributions $D:\01^n\rightarrow [0,1]$. The $(\varepsilon,\delta)$-sample complexity of learning $\Cc$ is the minimal sample complexity  over all $(\varepsilon,\delta)$-learners for $\Cc$.

	   We say $\A$ is a \emph{uniform-$(\varepsilon,\delta)$ learner} for a concept class $\Cc$ if the distribution $D$ is \emph{fixed} to be the uniform distribution over $\01^n$ and $\A$ learns $\Cc$ under the uniform distribution.

\subsubsection{Quantum PAC learning.}		
The quantum PAC model of learning was introduced by \cite{bshoutyandjackson:DNF}. In this model a quantum learning algorithm has access to a quantum computer and \emph{quantum  examples} $\sum_x\sqrt{D(x)}\ket{x,c(x)}$, and the goal is still output a {\em classical} hypothesis $h$ with the same requirements as in the classical setting. For every $\Cc$, the $(\varepsilon,\delta)$-quantum PAC complexities are defined as the quantum analogues to the classical complexity measures. For more on these learning models, we refer the reader to~\cite{arunachalam:quantumsurveylearning} and the references therein.

\paragraph{Noisy quantum examples.}
Following the work of Grilo et al.~\cite{grilo:LWEeasy}, we define noisy quantum PAC learning. Here, a learner is provided with copies of a {\em noisy quantum example} for a concept $c \in \Cc$ and distribution $D$ as a superposition of {\em noisy classical examples}.\footnote{In the classical setting, a noisy classical PAC learner obtains many $(x,c(x)+b_x)$ where $x$ is sampled from $D$ and $b_x$ is an independent random variable which equals $0$ with probability $1 - \eta$ and $1$ otherwise and using these noisy examples a learner needs to learn $c$.} Understanding the quantum and classical learnability of functions in the noisy setting is motivated by the connection to important problems in cryptography such as learning parity with noise~\cite{Pietrzak12} and learning with errors problem~\cite{regev:lwe}.  More concretely, a $\eta$-noisy quantum example for a concept $c$ is given~by
\begin{align}
\label{eq:noisy-samples-definition}
    \ket{\widehat{\psi}_c}=\sum_{x \in \{0,1\}^n} \sqrt{D(x)}\ket{x,c(x)\oplus b_x},
\end{align}
where each $b_x$ is an i.i.d.\ variable which equals $0$ with probability $1 - \eta$ and $1$ otherwise. Here again the goal of a learner is to learn a concept class $\Cc$ under all distributions $D$. The complexity of such learners is defined exactly as we defined it for quantum PAC learning, except that  we also allow the sample complexity of an $\eta$-noisy PAC learner to depend on the factor $1/(1-2\eta)$.\footnote{Note that when $\eta=1/2$, a learner is obtaining uniformly random bits of information in which case we cannot hope to learn $c$.}

\section{Quantum statistical query learning}	
\label{sec:defnqsqmodel}

	In this section, we introduce the model of quantum statistical query learning ($\Qsq$). We start by briefly describing the \emph{classical} $\Sq$ learning. 
  Let $\Cc\subseteq \{c:\01^n\rightarrow \pmset{} \}$ be a concept class. The goal of a statistical learning algorithm is to learn an unknown $c^*\in \Cc$ under an unknown distribution $D:\01^n\rightarrow [0,1]$. A (classical) $\Sq$ learning algorithm has access to a \emph{statistical query oracle} $\Stat$ which takes as input a \emph{tolerance} parameter $\tau \geq 0$ and a function $\phi:\01^n\times \pmset{}\rightarrow \pmset{}$ and returns a number $\alpha$ such that 
	\[  
	\Big |\alpha - \Exp_{x\sim D}[\phi(x,c^*(x))] \Big | \leq \tau\;.
	\]
	The $\Sq$ learning algorithm adaptively chooses a sequence $\{ (\phi_i,\tau_i) \}$, and based on the responses of $\{\Stat(\phi_i,\tau_i)\}_i$,  it outputs a hypothesis $h:\01^n\rightarrow \01$. We say that an $\Sq$ learning algorithm~$\A$ $\varepsilon$-learns $\Cc$ with query complexity $Q$ and tolerance $\tau$  if, for every $c^*\in \Cc$ and distribution~$D$,~$\A$ makes $Q$ classical $\Stat$ queries with tolerance at least $\tau$, and outputs a hypothesis $h$ that is $1-\eps$ close to~$c^*$ under $D$, i.e., $\Pr_{x\sim D}[h(x)=c^*(x)]\geq 1-\varepsilon$.\footnote{Note that in the $\Sq$ model, there is no ``$\delta$"-parameter associated to a learner, i.e., we require a $\Sq$ learner to always output a hypothesis $h$ that is $\varepsilon$-close to $c$ under $D$.}

    We extend this learning model to allow the algorithm to make {\em quantum statistical~queries}.
    \begin{definition}
    Let $\Cc\subseteq \{c:\01^n\rightarrow \01\}$ be a concept class and $D :\01^n\rightarrow [0,1]$ be a distribution. A \emph{quantum statistical query oracle} $\Qstat(M,\tau)$ for some $c^* \in \Cc$ receives as inputs  a tolerance parameter $\tau \geq 0$ and an \emph{observable} $M\in (\C^2)^{\otimes n+1} \times (\C^2)^{\otimes n+1}$ satisfying $\|M\|\leq 1$, and outputs a number~$\alpha$~satisfying
	$$
	\Big|\alpha- \langle \psi_{c^*}| M | \psi_{c^*}\rangle\Big|\leq \tau,
	$$
	where $\ket{\psi_{c^*}}=\sum_{x\in \01^n}\sqrt{D(x)}\ket{x,c^*(x)}$.
    \end{definition}
    	Observe that $\Qstat$ generalizes the classical $\Stat(\phi,\tau)$: if we take the diagonal matrix
	$$
	M=\sum_{z\in \01^n} \phi(z,c(z))\ketbra{z,c(z)}{z,c(z)},
	$$
	then $\Qstat(M,\tau)$ outputs a number $\alpha \in \R$ that is $\tau$-close to $\Exp_{x\sim D} [\phi(x,c(x))]$, as in the classical case. Allowing $M$ to be an arbitrary quantum observable lets the $\Qsq$ learning algorithm to acquire a broader range of statistics from the coherent superposition of labeled examples.

    \begin{definition}
        Let $\Cc$ be a concept class and $D :\01^n\rightarrow [0,1]$ be a distribution.  We stay that~$\Cc$ can be $\eps$-learned in the quantum statistical query model with $Q$ queries, if there is an algorithm~$\A$ such that for every $c^* \in \Cc$, $\A$ makes at most $Q$ $\Qstat$ queries and outputs a hypothesis~$h$~satisfying 
$\Pr_{x\sim D}[h(x)\neq c^*(x)]\leq \varepsilon$.
\end{definition}

 We justify this model as follows. In the classical case, one can think of the input $\phi$ to the $\Stat$ oracle as a specification of a \emph{statistic} about the distribution of examples $(x,c^*(x))$, and the output of the $\Stat$ oracle is an  {\em estimation} of $\phi$: one can imagine that the oracle receives i.i.d.~labeled examples $(x,c^*(x))$ and empirically computes an estimate of $\phi$, which is then forwarded to the learning algorithm.	In the quantum setting, one can imagine the analogous situation where the oracle receives copies of the quantum example state $\ket{\psi_{c^*}}$, and performs a measurement  indicated by the {\em observable} $M$ on each copy and outputs an estimate of $\langle \psi_{c^*}| M | \psi_{c^*}\rangle$.  We emphasize that the learning algorithm is still a \emph{classical} randomized algorithm and only receives statistical estimates of measurements on quantum~examples.
   Similar to the PAC setting, we are also interested in the sample and time complexity of learning concept classes in the quantum statistical model.
    
    \begin{definition}
    Let $\Cc$ be a concept class and $D :\01^n\rightarrow [0,1]$ be a distribution. 
    We define $\Qsq_{\eps}(c,D)$ as the minimal number of $\Qstat$ queries that a learner $\A$ needs to make to $\eps$-learn~$c$. We  define the \emph{statistical query complexity of $\Cc$} as
	$$
	\Qsq_{\eps}(\Cc)=\max_{c\in \Cc}\max_{D} \Qsq_{\eps}(c,D).
	$$
We say that $\Cc$ can be $\varepsilon$-learned in  polynomial time (polynomial with respect to the precision $1/\tau$, the error parameter $1/\varepsilon$ and the description size of $\Cc$) under the distribution $D$ in the $\Qsq$ model if there is a polynomial time  algorithm $\A$ that $\eps$-learns $\Cc$ under $D$. We say that $\Cc$ can be $\varepsilon$-learned in  polynomial time if the learning algorithm $\mathcal{A}$ works for every distribution $D:\01^n\rightarrow [0,1]$.\footnote{{Here, by polynomial-time algorithm, we mean the number of gates used in the quantum algorithm is polynomial in the relevant parameters.}}
    \end{definition}

When the bias is not explicitly mentioned, $\Qsq(\Cc)$
denotes the statistical query complexity of learning $\Cc$ with bias $\varepsilon=1/3$. 

Like in the classical case, our first observation is that if there exists an efficient quantum statistical learning algorithm using $\Qstat$ queries then there also exists a  quantum learning algorithm in the standard and noisy PAC setting.

	\begin{theorem}
		\label{claim:kearnssqimpliesnoisy}
		Let $\Cc$ be a concept class. Let $\tau,\delta>0$ and $\eta< \max\{1/2,2\tau^2\}$. Suppose there exists an $\varepsilon$-$\Qsq$ algorithm that makes $Q$ $\Qstat$ queries with tolerance at least~$\tau$. Then,
		\begin{enumerate}
			\item 
			there exists a $(\varepsilon,\delta)$-quantum PAC learner for $\Cc$ that uses $O(\tau^{-2} Q\log (Q/\delta))$ quantum examples.
			\item there exists a $\eta$-noisy $(\varepsilon,\delta)$-quantum PAC learner for $\Cc$ that uses $O((\tau-\sqrt{\eta})^{-2} Q\log (Q/\delta))$ many $\eta$-noisy quantum examples.
		\end{enumerate} 
	\end{theorem}
	
	\begin{proof}
    Suppose there exists a quantum statistical algorithm that makes the $Q$ $\Qstat$ queries $\{(M_1,\tau),\ldots,(M_Q,\tau)\}$ and the output of $\Qstat$ queries are $\{\alpha_1,\ldots,\alpha_Q\}$, where
		$$
			\Big\vert \alpha_i -  \langle \psi_c \vert M_i \vert \psi_c \rangle \Big\vert \leq \tau,
		$$
		and $\ket{\psi_c}=\sum_x\sqrt{D(x)}\ket{x,c(x)}$.\footnote{We remark that we consider non-adaptive queries for simplicity and that our argument works even if the $\Qsq$ learner makes $Q$ \emph{adaptive} $\Qstat$ queries.} 		We now prove the first statement in the theorem.  Consider a quantum PAC learner that does the following: for every $i\in [Q]$, a quantum learner obtains $T=\log (Q/\delta)/{\tau}^2$ copies of $\ket{\psi_c}$ and measures each of them according to $M_i$ with outcomes $\{a^i_1,\ldots,a^i_{T}\}\in [-1,1]^T$. The quantum PAC learner simply passes $\beta_i=\frac{1}{T}\sum_{j=1}^{T}a^i_j$ to the $\Qsq$ learner. By a Chernoff bound, observe that
		$$
		\Pr\Big[\Big|\beta_i - \langle \psi_c \vert M_i \vert \psi_c \rangle\Big|\leq \tau\Big]\geq 1-\frac{\delta}{Q} \qquad \text{ for every } i\in [Q],
		$$
		where the probability is taken over the randomness in measurement. By the union bound over all~$Q$, with probability at least $1-\delta$, the $\Qsq$ learner obtains $\{\beta_1,\ldots,\beta_Q\}$ which are the response~$Q$ $\Qstat$ queries up to precision $\tau$. Hence the $\Qsq$ learner (and the quantum PAC learner) outputs a hypothesis $h$ that satisfies $\Pr_{x\sim D}[h(x)=c(x)]\geq 2/3$.  The total number of quantum examples used by quantum PAC learner is $ Q\cdot \log (Q/\delta)/{\tau}^2$.
		
		We now prove the second statement in the theorem. Consider the case where an $\eta$-noisy quantum PAC learner is given $T=(\tau-\sqrt{\eta/2})^{-2}\log(Q/\delta)$  noisy quantum examples $\bigotimes_{j \in [T]} \ket{\widehat{\psi}_j}$ for each query $i \in [Q]$, where each $\ket{\widehat{\psi}_j}$ is a fresh noisy example as described in \Cref{eq:noisy-samples-definition}. 
		The noisy PAC learner behaves like a standard PAC learner, for every $i\in [Q]$, the learner measures each of the $T$ copies according to $M_i$, obtains $\{\alpha^i_1,\ldots,\alpha^i_T\}$ and passes $\beta'_i=\frac{1}{T}\sum_{j=1}^Ta^i_j$ to the $\Qsq$ learner. Before we analyze the difference between $\beta'_i$ and $\langle \widehat{\psi}_c | M_i| \widehat{\psi}_c\rangle$, we first observe that for every $M_i$ satisfying $\|M_i\|\leq 1$ and for every $j\in [T]$, we~have
		\begin{align*}
		   \Exp \Big[ \big|\langle \psi_c | M_i|\psi_c  \rangle-\langle \widehat{\psi}_j | M_i| \widehat{\psi}_j\rangle \big|\Big]\leq \Exp\Big[\big\|\ketbra{\widehat{\psi}_j}{\widehat{\psi}_j}-\ketbra{\psi_c}{\psi_c}\big\|\Big]=\Exp\Big[ \sqrt{1-\langle \widehat{\psi}_j| \psi_c\rangle}\Big]=\sqrt{1-\sqrt{1-\eta}}\leq \sqrt{\eta},
		\end{align*}
		where we use the definition of trace distance in the first inequality and $\sqrt{1-x}\leq x/2$ for $x\leq 1$ in the last inequality. Hence for every $i,j$, we have
		\begin{align} \label{eq:def-mu}
		    \mu := 		\Exp \Big[ \langle \widehat{\psi}_j | M_i| \widehat{\psi}_j\rangle \Big]
		 \in \left[\bra{\psi_c}M_i\ket{\psi_c} - \sqrt{\eta}, \bra{\psi_c}M_i\ket{\psi_c} + \sqrt{\eta}\right],
		\end{align} 
		and such value is \emph{independent} of $j$.  Using a Chernoff bound over the $T$ noisy quantum examples, we have
		\begin{align}
    \label{eq:howfarpsietafrompsi}
		\Pr\Big[\Big|\beta'_i - \mu \Big|\leq \tau-\sqrt{\eta}\Big]\geq 1-\frac{\delta}{Q},
		\end{align}
		where the probability is taken over the randomness in measurement. In particular, with probability $1 - \frac{\delta}{Q}$ we have that
		\begin{align}
\Big|\beta'_i - \langle \psi_c \vert M_i \vert \psi_c \rangle\Big| \leq \Big|\beta'_i- \mu \Big|+\Big| \mu - \langle \psi_c \vert M_i \vert \psi_c \rangle\Big|\leq \tau-\sqrt{\eta}+\sqrt{\eta}=\tau,
		\end{align}
		where the first inequality used the triangle inequality and the last inequality used \Cref{eq:def-mu,eq:howfarpsietafrompsi}. We now use the same argument as the PAC setting to argue that with probability at least $1-\delta$, a $\Qsq$ learner which obtains $\{\beta'_1,\ldots,\beta'_Q\}$ will output $h$ that satisfies $\Pr_{x\sim D} [h(x)=c(x)]\geq 2/3$, hence proving the theorem statement. The total number of quantum examples used by quantum PAC learner is  $O((\tau-\sqrt{\eta})^{-2} Q\log (Q/\delta))$.
	\end{proof}

	\section{Learning concept classes quantum efficiently} 
	\label{sec:introlearneff}
In this section, we show how to quantum-efficiently learn concept classes in the $\Qsq$ model that are provably hard to learn in the classical $\Sq$ model.
Our key technical tool that will lead to such learning algorithms is a procedure to estimate the Fourier mass of a concept $c$ on a subset of $\01^n$ using a single $\Qstat$.

\begin{lemma}
\label{lem:stattoestimateinf}
Let $f:\01^n\rightarrow \pmset{}$ and $\ket{\psi_f}=\frac{1}{\sqrt{2^n}}\sum_x\ket{x,f(x)}$.
 There is a procedure that on input  $T \subseteq \01^n$,  outputs a $\tau$-estimate of $\sum_{S\in T}\widehat{f}(S)^2$ using one $\Qstat$ query with tolerance $\tau$.
 \end{lemma}
	
\begin{proof}
Let $M=\sum_{S\in T}\ketbra{S}{S}$. The observable used in $\Qstat$ is then
		$$
    M'=\mathsf{H}^{\otimes (n+1)} \cdot \Big( \id^{\otimes n}\otimes \ketbra{1}{1} \Big) \cdot M\cdot \Big( \id^{\otimes n}\otimes \ketbra{1}{1}\Big) \cdot \mathsf{H}^{\otimes (n+1)}.
		$$ 
    Operationally, $M'$ corresponds to first apply the
    Fourier transform on $\ket{\psi_f}$, {post-selecting on the last qubit being $1$} and finally applying $M$ to the first $n$ qubits.
 In order to see the action of $M'$ on~$\ket{\psi_f}$, first observe that $\mathsf{H}^{\otimes (n+1)}\ket{\psi_f}$ yields		
		$$
		\frac{1}{\sqrt{2^{n}}}\sum_x\ket{x,f(x)}\rightarrow \frac{1}{2^n}\sum_{x,y}\sum_{b\in \01} (-1)^{x\cdot y+b\cdot f(x)}\ket{y,b}.
		$$
        Conditioned on the $(n+1)$-th qubit being $1$, we have that the resulting quantum state is $\ket{\psi'_f}=\sum_Q \widehat{f}(Q)\ket{Q}$. Applying $M$ to the resulting state gives us
		$$
		\langle\psi'_f\vert M\vert \psi'_f\rangle=\sum_{\substack{R,Q\in \01^n\\S \in  T}}\braket{S}{Q}\braket{S}{R}\widehat{f}(R)\widehat{f}(Q)=\sum_{S\in T}\widehat{f}(S)^2.
		$$
Therefore, one $\Qstat(M',\tau)$ query results in a $\tau$-approximation of $\sum_{S\in T}\widehat{f}(S)^2$. 
\end{proof}	
	
We now use this lemma to show efficient $\Qsq$ learners for parities, $k$-juntas and DNFs.

\subsection{Learning Parities}

We start by showing a  polynomial time $\Qsq$ learner for parities. Classically, Kearns \cite{kearns:statistical} showed that $2^{\Omega(n)}$ $\Stat$ queries (with tolerance at least $2^{-\Omega(n)})$ are necessary to weakly learn parities under the uniform distribution. 

\begin{lemma}
\label{lem:parities}
		\label{lem:quantumsqparity}
		The concept class $\Cc=\{\chi_s:\01^n\rightarrow \01 \}_s$ of parities can be exactly learned 
		with $O(n)$ $\Qstat$ queries with tolerance at least $1/3$  under the uniform distribution. 
\end{lemma}
	\begin{proof}
	   Let $c:\01^n\rightarrow \pmset{}$ be a parity function defined as $c(x)=(-1)^{s\cdot x}$ for an unknown $s\in \01^n$. Then it is not hard to see that $\Inf_i(c)=1$ for all $i\in \supp(s)$ and $\Inf_i(c)=0$ otherwise. Using this observation, the $O(n)$ query quantum algorithm is straightforward: for every $i\in [n]$, we use Lemma~\ref{lem:stattoestimateinf} to estimate the $i$th influence using one $\Qstat$ query with tolerance $\tau=1/3$ (we let $T=\{S\subseteq \01^n:s_i=1\}$ in Lemma~\ref{lem:stattoestimateinf}, in which case we have $\sum_{S\in T}\widehat{f}(S)^2=\Inf_i(f)$). Suppose $\Inf_i(f)=1$, then  the outcome of the $i$th $\Qstat$ query is in the interval $[2/3,4/3]$ and if $\Inf_i(f)=0$, the outcome of the  $\Qstat$ query is in the interval $[-1/3,1/3]$. Given the outcomes of the queries, a quantum learning algorithm can easily learn $s \in \01^n$, and hence~$c$ exactly. {In order to understand why this algorithm can be implemented efficiently observe that for $T=\{S\subseteq \01^n:s_i=1\}$, the corresponding $M$ we need to implement in Lemma~\ref{lem:stattoestimateinf} can be written as 
	   \begin{align}
	       \label{eq:defnofM}
	   M=\sum_{S\in T}\ketbra{S}{S}= \ketbra{1}{1}_i\otimes \mathsf{H}^{\otimes(n-1)}\cdot\Big(\ketbra{0}{0}^{\otimes [n]\backslash \{i\}}\Big)\cdot \mathsf{H}^{\otimes(n-1)},
	   	   \end{align}
	   where the $i$th qubit is fixed to $\ketbra{1}{1}$ and the remaining $n-1$ qubits (excluding the $i$th qubit) can be obtained by applying the Hadamard transform on the $n-1$ qubits. Since $M$ can be implemented using $\poly(n)$ gates, one can learn parities quantum efficiently. }
	 	\end{proof}

\subsection{Learning $k$-juntas}
Using a similar idea to the parities problem, we show how to efficiently learn the class of $O(\log n)$-juntas under the uniform distribution in polynomial time in the $\Qsq$ model.   The idea of this quantum learning algorithm is to first learn most of $k$ influential variables of a junta using Lemma~\ref{lem:stattoestimateinf} and then  approximates all the $2^k$ Fourier coefficients of the function using $2^k$ classical $\Stat$ queries with tolerance $2^{-k}$.

\newcommand{\textlemjuntas}{
		Let $\Cc$ be the concept class of $k$-juntas. Then, $O(n+2^{O(k)})$ many $\Qstat$ queries  with tolerance at least $O(\varepsilon\cdot 2^{-k/2})$ suffices to $\varepsilon$-learn $\Cc$ under the uniform distribution.
}
	\begin{lemma}
	\label{lem:juntas}
	\textlemjuntas
 	\end{lemma}
	\begin{proof}
		Our algorithm is divided into two steps: first, we use $O(n)$ $\Qstat$ queries with tolerance $O(\varepsilon/k)$ to learn the variables $i\in [n]$ for which $\Inf_i(f)\geq \varepsilon/k$. Let $T$ be the set of such variables.  %
		Next we use classical statistical queries to estimate all the Fourier coefficients $\widehat{f}(V)$ for every $V\subseteq T$ using at most $O(2^k)$ $\Qstat$ queries with tolerance $2^{-\Omega(k)}$.

    Let $c$ be a $k$-junta over the variables in $Q\subseteq [n]$ with size $|Q|=k$, i.e., $c(x)=f(x_Q)$ for some arbitrary $f:\01^k\rightarrow \01$. Then, it is not hard to see that for all $i\notin Q$, we have $\Inf_i(c)=0$. 
    Since the goal of a quantum statistical learner is to $\varepsilon$-learn $c$, it suffices to obtain the variables $i\in [n]$ whose influence $\Inf_i(f)$ is at least $\varepsilon/2k$.
    Our quantum learning algorithm proceeds as follows: for every $i\in [n]$, we use Lemma~\ref{lem:stattoestimateinf} to estimate $\Inf_i(f)$ upto precision $\varepsilon/(5k)$. Suppose the outcome of these queries is $\alpha_1,\ldots \alpha_n$, we let 
		$$
		T=\Big\{i\in [n]: \alpha_i \geq \varepsilon/4k \Big\}.
		$$
	First, notice that $T\subseteq Q$, since $\Inf_i(f) = 0$ implies that $\alpha_i \leq \varepsilon/5k$. Observe also that for every $i\in Q \setminus T$, we have $\Inf_i(f) < \varepsilon/(2k)$: in order to see this, suppose $\Inf_i(f) \geq \varepsilon/2k$ for some $i\in Q\setminus T$, then the $\Qstat$ query to estimate $\Inf_i(f)$ would produce an $\alpha_i$ such that 
		$$
		\alpha_i \geq \Inf_i(f)-\varepsilon/(5k) \geq \varepsilon/4k,
		$$
    but this contradicts the fact that $i\not\in T$. Hence $T$ has captured all the variables with high influences. In particular
		\begin{align}
		\label{eq:influenceissmall}
        \sum_{i\in [n]\backslash T}\Inf_i(c)=\sum_{i\in Q\backslash T}\Inf_i(f)\leq k\cdot \frac{\varepsilon}{2k}=\varepsilon/2,
		\end{align}
		where the first equality used the fact that $\Inf_i(c)=0$ for every $i\notin Q$ and the inequality used that there are at most $k$ influential variables, hence $|Q|\leq k$

		In the second phase, we $\varepsilon$-approximately learn the junta. In order to do this, we estimate the Fourier coefficients for  all subsets of  $T$. For every $V\subseteq T$, we make one $\Qstat$ query to approximate  $\widehat{f}(V)$ upto error $\sqrt{\varepsilon/2}\cdot 2^{-k/2}$: for $V\subseteq [n]$, let $\phi(x,b)= b\cdot (-1)^{V\cdot x}$ for all $x\in \01^n,b\in \01$, hence $\Exp_x [\phi(x,c(x))]=\Exp_x \big[c(x)\cdot (-1)^{V\cdot x}\big]=\widehat{c}(V)$.
		Overall, it takes $2^{|V|}\leq 2^k$ many $\Qstat$ queries to estimate all Fourier coefficients $\{\widehat{c}(V):V\subseteq T\}$. Once we obtain all these approximations $\{\alpha_V\}_{V\subseteq Q}$, we output the function $$
		g(x)=\sign \Big(\underbrace{\sum_{V\subseteq T} \alpha_V \cdot  \chi_V(x)}_{:=h(x)}\Big), \quad \text{ for every } x\in \01^{n}.
		$$
		We now argue that $g$ is in fact $\varepsilon$-close to $c$:
	\begin{alignat}{2}
				\label{eq:gclosetocjunta}
				\begin{aligned}
          \Pr_{x\in \01^n}[c(x)\neq g(x)]&=\Exp_{x}[c(x)\neq \sign(h(x))]\\ &\leq \Exp_x[|c(x)-h(x)|^2]\\
		 &=\sum_{V}(\widehat{h}(V)-\widehat{c}(V))^2
		\begin{aligned}[t]
    	      &= \sum_{V\subseteq T} (\alpha_V-\widehat{c}(V))^2 +\sum_{V \subseteq [n]\backslash T} \widehat{c}(V)^2\\
		  & \leq 2^k\cdot \frac{\varepsilon}{2^{k+1}}+\sum_{i\in [n]\backslash T}\Inf_i(c)\leq \varepsilon,
         \end{aligned}
         \end{aligned}
		\end{alignat}
		where $[\cdot]$ is the indicator of an event, the second equality used Plancherel's identity to conclude $\Exp_x (c(x)-h(x))^2=\sum_V(\widehat{c}(V)-\widehat{h}(V))^2$, the second inequality follows by definition of $\Inf_i(c)=\sum_{S\subseteq [n]: S\ni i}\widehat{c}(S)^2$ and the fact that $\Qstat$ queries return $\alpha_V$ which are a $\varepsilon/2\cdot 2^{-k/2}$ approximation of~$\widehat{c}(V)$, and finally the last inequality used Eq.~\eqref{eq:influenceissmall}. {For the same reason as in Lemma~\ref{lem:parities}, phase~$1$ can be performed quantum-efficiently (since the $M$s can be expressed as Eq.~\eqref{eq:defnofM} which takes $\poly(n)$ gates to implement) and phase $2$ takes time polynomial in $n,2^k$ since each $\phi$ can be computed in time $O(n)$ and we make  $2^k$ many  $\Stat$ queries.}
	\end{proof}

  Notice that classically, every $\Sq$ learner for $k$-juntas needs to make $n^{\Omega(k)}$ $\Stat$ queries with tolerance at least $n^{-\Omega(k)}$, since this class contains at least $\binom{n}{k}$ distinct parity functions.

\subsection{ Learning Disjunctive Normal Forms (DNFs)}	Finally, we give a polynomial time learning for learning $\poly(n)$-sized DNFs in the $\Qsq$ model. Classically we need $n^{\Omega(\log n)}$ classical $\Stat$ queries (with tolerance $n^{-\Omega(\log n)}$)  to learn DNFs (since $\poly(n)$-sized DNFs contain $O(\log n)$-juntas which in turn contain at least $n^{O(\log n)}$ distinct parity~functions).

    	The key step of the proof is to replace the membership queries in the well-known Goldreich-Levin ($\textsf{GL}$) algorithm~\cite{goldreich:hardcore,kushimansour:GL} by quantum statistical queries. In particular, for a function $c:\01^n\rightarrow \pmset{}$, our ``quantum statistical" $\textsf{GL}$ algorithm makes $\poly(n,1/\tau)$ $\Qstat$ queries with tolerance at least $\tau$ and returns a set $U=\{T_1,\ldots,T_\ell\}\subseteq [n]$ such that if $|\widehat{c}(T)|\geq \tau$, then $T\in U$, and if $T\in U$, we have $|\widehat{c}(T)|\geq \tau/2$. Using this subroutine, for an $s$-term DNF we can find all the  Fourier coefficients which satisfy $|\widehat{c}(T)|\geq 1/s$ using $\poly(n,s)$ many $\Qstat$ queries. After this, one can use the classical algorithm for DNF learning by \cite{Feldman12:dnfwithoutboosting} in order to approximate the $s$-term DNF. Overall our quantum statistical oracle uses $\poly(n)$ $\Qstat$ queries of tolerance $1/\poly(n)$ to learn $\poly(n)$-sized DNF formulas.

	In order to prove the main lemma statement, we first argue that,  in the classical Goldreich-Levin  algorithm (\textsf{GL} algorithm)~\cite{goldreich:hardcore,kushimansour:GL}, we can replace \emph{classical membership queries} by \emph{quantum statistical queries}. 
	
	\begin{theorem}[Goldreich-Levin theorem using $\Qstat$ queries]
		\label{thm:goldreichlevin} Let $f:\pmset{n}\rightarrow \pmset{}$,  $\tau \in (0,1]$. There exists a $\poly(n,1/\tau)$-time quantum statistical learning algorithm that with high probability outputs $U=\{T_1,\ldots,T_\ell\}\subseteq [n]$ such that: (i) if $|\widehat{f}(T)|\geq \tau$, then $T\in U$; and (ii) if $T\in U$, then $|\widehat{f}(T)|\geq \tau/2$.	
	\end{theorem}
	We do not prove Theorem~\ref{thm:goldreichlevin} here since it follows the classical \textsf{GL} algorithm almost exactly.\footnote{An interested reader is referred to Section $3.5$ of \cite{donnell:book} for details.}  Instead, we only state the difference between the proofs of classical \textsf{GL} algorithm and Theorem~\ref{thm:goldreichlevin}. In classical \textsf{GL} algorithm, one uses classical membership queries to perform the following task   in time $\poly(n,1/\varepsilon)$: let $Q\subseteq [n]$, for every $S\subseteq Q$, estimate $\sum_{V\subseteq {Q}^c}\widehat{c}(S\cup V)^2$ upto precision $\varepsilon$. Instead, in the proof of Theorem~\ref{thm:goldreichlevin}, we simply use Lemma~\ref{lem:stattoestimateinf} to estimate $\sum_{V\subseteq {Q}^c}\widehat{c}(S\cup V)^2$ using one quantum statistical query $\Qstat(c,\varepsilon)$ (we use Lemma~\ref{lem:stattoestimateinf} by setting $T=\{(S\cup V):V\subseteq  Q^c\}$ for a fixed $S$). The remaining part of \textsf{GL} algorithm, as well as the proof of Theorem~\ref{thm:goldreichlevin}, does not involve membership queries to $c$. {Observe again that for a fixed $S$ and $T=\{(S\cup V):V\subseteq  Q^c\}$, we can write~$M$ in Lemma~\ref{lem:stattoestimateinf} as 
		 $$
		 M=\sum_{R\in T}\ketbra{R}{R}= \ketbra{1}{1}_S\otimes \ketbra{0}{0}_{Q\backslash S}\otimes \mathsf{H}^{\otimes(n-|Q|)}\cdot\Big(\ketbra{0}{0}^{\otimes [n]\backslash Q}\Big)\cdot \mathsf{H}^{\otimes(n-|Q|)},
	$$
	   wherein for all $i\in S$, we fix the $i$th qubit to $\ketbra{1}{1}$, for $j\in Q\backslash S$ we set the  $j$th qubit to be $\ketbra{0}{0}$ and the remaining $n-|Q|$ qubits can be obtained by applying the Hadamard transform on the $n-|Q|$ qubits. Clearly such $M$s can be implemented quantum efficiently using $\poly(n)$ gates.
	We now prove our main lemma using Theorem~\ref{thm:goldreichlevin}.}
	
	\begin{lemma}
	\label{lem:learnDNF}
		Let $\Cc$ be the concept class of $\poly(n)$-sized DNFs. Then there exists a $\poly(n)$-query $\Qsq$ algorithm that $\varepsilon$-learns $\Cc$ under the uniform distribution.
 	\end{lemma}
	\begin{proof}
        Our quantum DNF learning algorithm follows the same ideas as the classical DNF learning by Feldman~\cite{Feldman12:dnfwithoutboosting}. We simply replace the classical membership queries in Feldman's algorithm by quantum statistical  queries.\footnote{Alternatively, we could have used the  weak-quantum learning for DNFs from \cite{bshoutyandjackson:DNF}, followed by the statistical query \emph{boosting} algorithm by \cite{aslam:boosting}. %
        } The only use of membership queries in Feldman's algorithm (in particular, Corollary $5.1$ of \cite{Feldman12:dnfwithoutboosting}) is to run \textsf{GL} algorithm to collect all the ``large" Fourier coefficients of low Hamming weight: i.e., for an $s$-term DNF $c$, Feldman's learning algorithm uses membership queries to find all the $S$s that satisfy  $|\widehat{c}(S)|\geq \Omega(\varepsilon/s)$. In order to collect such $S$s, we use \textsf{GL} (see Theorem~\ref{thm:goldreichlevin}) which makes $\poly(n,s/\varepsilon)$ quantum statistical queries to find all the heavy Fourier coefficients of $c$ and discard those coefficients with large Hamming weight. The remaining part of the Feldman's algorithm in order to $\varepsilon$-learn $c$ does not require membership queries to $c$ and our quantum learner simply continues with Feldman's algorithm. The overall running time of Feldman's algorithm and our quantum learning algorithm is $\poly(n,s/\varepsilon)=\poly(n/\varepsilon)$ since we are concerned with $S=\poly(n)$-sized DNFs. 
        	\end{proof}

\section{Statistical query dimension}\label{sec:introlowerbounds}
In a seminal work, Blumer et al.~\cite{blumer:vc} showed that  sample complexity of PAC learning is characterized by a combinatorial parameter called $\VC$ dimension (which was defined by Vapnik and Chervonenkis~\cite{vapnik:vcdimension}). Similarly, Blum et al.~\cite{blum:sqdim} introduced a combinatorial parameter called \emph{statistical query dimension} that characterizes the sample complexity of weak  statistical query learning.\footnote{Here, ``weak" refers to the fact that the output hypothesis $h$ of the learner needs to weak-approximate the target concept $c^*$ under the unknown distribution $D$, i.e., $\Pr_{x\sim D}[h(x)=c^*(x)]\geq 1/2+1/\poly(n)$.} 	Roughly the statistical query dimension for a concept class $\Cc$ and distribution $D$ measures the maximum number of concepts in $\Cc$ that are nearly uncorrelated with respect to $D$.  Let us define it more formally.
	\begin{definition}
		Let $\Cc$ be a concept class and $D$ be a distribution. Then $\WSQDIM(\Cc,D)$ is defined as the largest $d$ such that there exists $\{c_1,\ldots,c_d\}\subseteq \Cc$ such that $|\Exp_{x\sim D}[c_i(x)\cdot c_j(x)]|\leq \frac{1}{d}$ for every $i\neq j$. We define $\WSQDIM(\Cc)=\max_D \{\WSQDIM(\Cc,D)\}$.
	\end{definition}
	
	Using this combinatorial quantity, Blum et al.~\cite{blum:sqdim} showed the following characterization.
		\begin{theorem}[Blum et al.~\cite{blum:sqdim}]
		\label{thm:weak_sq}
		Let $\Cc\subseteq \{c:\01^n\rightarrow \01\}$ be a concept class and $D$ be a distribution. Suppose $\WSQDIM(\Cc,D)=~d$.
		\begin{itemize}
			\item There exists a $\Sq$ algorithm for learning $\Cc$ under $D$, with error $\frac{1}{2} - \frac{1}{3d}$, that makes $d$ $\Stat$ queries each with tolerance at least $1/(3d)$.
			\item If all $\Stat$ queries are made with tolerance $\geq d^{-1/3}$, then at least $d^{1/3}$ queries to the $\Stat$ oracle is necessary in order to weakly $\Sq$ learn $\Cc$.
		\end{itemize}
	\end{theorem}
	Subsequently, there have been many works~\cite{szorenyi2009characterizing,yang:sqdimension,feldman:completechar} that generalized and strengthened $\WSQDIM(\Cc,D)$ in order to characterize other variants of $\Sq$ learning model. We do not define these strengthened combinatorial parameters since we will not be using them.

We  now show that that for every concept class $\Cc$, distribution $D$ and tolerance $\tau>0$, every learner needs to make $\log_{1/\tau}(\WSQDIM(\Cc,D))$ many $\Qstat$ queries with tolerance at least $\tau$ in order to learn $\Cc$ under $D$ with error at most $1/2-1/d$.

	\begin{lemma}\label{lem:lb-qsqdim}
		Let $\tau>0$. Let $\Cc\subseteq \{c:\01^n\rightarrow \pmset{}\}$ and $D:\01^n\rightarrow [0,1]$ be a distribution such that  $\WSQDIM(\Cc,D)=d$. Then, every weak $\Qsq$ learning algorithm for $\Cc$ (with error at most $1/2- 1/d$) needs to make $\Omega(\log_{1/\tau} d)$ $\Qstat$ queries each of tolerance at least $\tau$. Moreover, this lower bound is tight for the class of parity functions on $n$ bits.
	\end{lemma}
\	\begin{proof}
		In order to prove the lemma, we will use the following simple fact: suppose we place $d$ points on the unit interval $[-1,1]$, then one can always find a $2\tau$-sized ball within the unit interval that covers $\tau d$ points. In order to see this: suppose by contradiction, assume that \emph{every} $2\tau$-sized ball in the interval $[-1,1]$ covers strictly lesser than $\tau d$ points. Then place $1/\tau$ $2\tau$-sized balls to cover $[-1,1]$. By assumption, since each of the $1/\tau$ $2\tau$-sized covers $< \tau d$ points, the total number of points in the interval $[-1,1]$ is strictly lesser than $d$, which contradicts the original assumption that we placed $d$ points in the interval~$[-1,1]$.
		
		Let $\Cc$ be a concept class and $D$ be a distribution satisfying $\WSQDIM(\Cc,D)=d$. By definition, there are $d$ concepts $\Cc'=\{c_1,\ldots,c_d\}$ such that for every $c_i\neq c_j \in \Cc'$, we have $\Big|\Exp_{x\sim D} [c_i(x)\cdot c_j(x)]\Big|\leq 1/d$.
		We now show that every $\Qsq$ learner for $\Cc$ with bias $1/d$ and tolerance at least $2\tau$, needs to make $\Omega(\log_{1/\tau} d)$ quantum statistical queries. The proof goes via an adversarial argument, i.e., we show how the replies of a $\Qstat$ oracle can  enforce a $\Qsq$ learner to make $\Omega(\log_{1/\tau} d)$ queries to the $\Qstat$ oracle. 
		
    Suppose the first $\Qstat$ query made by the learner is specified by the operator $M_1$ and precision $2\tau$. The adversarial $\Qstat(M_1,2\tau)$ oracle computes $A_1 = \big(\langle {\psi_{c}}\vert M_1 \vert  {\psi_{c}}\rangle \big)_{c \in \Cc'}$, which contains $d$ values. By the argument in the beginning of the proof, there exists a point $x_1$ such that at least $\tau d$ points in $A_1$ lie within a $2\tau$-radius of $x_1$.
    Then, the $\Qstat(M_1,2\tau)$ oracle responds with $x_1$. This narrows down the search space for the learner to at least~$\tau d$ candidate concepts $\Cc_1\subseteq \Cc$. Suppose the next quantum statistical query of the learner is $(M_2,2\tau)$, the $\Qstat(M_2,2\tau)$ oracle  computes     the sequence
    $A_2=\big(\langle {\psi_{c}}\vert M_2 \vert {\psi_{c}}\rangle\big)_{c \in \Cc_1}$
    with at least $\tau d$ values and responds with $x_2$ such that there are at least $\tau^2 d$ points in $A_2$ around $x_2$. This process repeats for all the $T$ $\Qstat$ queries $\{(M_i,2\tau):i\in [T]\}$ made by the~$\Qsq$~learner.
		
    Suppose $T< \log_{1/(2\tau)} d$ queries. Then after making $T$ queries, there are at least two distinct concepts $c_1,c_2\in\Cc'$ that satisfy 
    \begin{align}
        \label{eq:dist-concepts}
\Pr_{x\sim D}[c_1(x)\neq c_2(x)]\geq 1/2-1/2d   
    \end{align} and these concepts are consistent with all the $T$ quantum statistical queries queries made so far. Let $h$ be the output of the quantum learner. Given Eq.~\eqref{eq:dist-concepts}, it must be the case that either 
    $\Pr_{x\sim D}[c_1(x)\neq h(x)]\geq 1/2-1/4d$ or $\Pr_{x\sim D}[c_2(x)\neq h(x)]\geq 1/2-1/4d$, and we can choose, adversarially, the concept that maximizes such distance. Hence $T$, the number of queries made by an $\Qsq$ learner, needs to be at least $\log_{1/(2\tau)} d$, proving the lower bound.

    We now show that this lower bound is tight. Let $\PARITY_n$ be the class of parity functions on $n$ bits. For $c,c' \in \PARITY_n$ with $c \neq c'$, under the uniform distribution $\uniform_n$ we have that $\Exp_{x \sim \uniform_n} [c(x) c'(x)] = 0$. Since $|\PARITY_n| = 2^n$, we have $\WSQDIM(\PARITY,\uniform_n)=2^n$. Along with Lemma~\ref{lem:quantumsqparity}, the lower bound above is tight for $\PARITY_n$ under the uniform distribution.
    \end{proof}
    
\
\subsection{Connections to communication complexity}\label{sec:introcc}

We now present connections between the weak statistical query dimension and communication complexity.  Several works~\cite{kremer:CC,jainzhang:CC,ambainis:superdense} showed a surprising connection between communication complexity and learning theory: for every~$F:\01^n \times \01^n\rightarrow \01$, the classical and quantum one-way communication complexities of $F$ (under product distributions) are characterized by the $\VC$ dimension of the concept class $\Cc_F=\{F_x:\01^n\rightarrow \01: F_x(y)=F(x,y)\}_{x\in \01^n}$. We now prove that weak statistical query dimension of $\Cc_F$ also lower bounds the complexity in this communication model when $\varepsilon$ asymptotically goes to zero.

We now define the model formally. Let $\Cc\subseteq \{c:\01^n\rightarrow \01\}$ and  $\mu:\Cc\times \01^n\rightarrow [0,1]$ be a \emph{product distribution}. We consider the following task: $(c,x)$ are picked from $\Cc\times \01^n$ according to~$\mu$, and Alice is given as input $c\in \Cc$  and Bob is given $x\in \01^n$. Alice and Bob share random bits and Alice is allowed to send classical bits to Bob, who needs to output~$c(x)$ with probability $1/2+\gamma$. We let $\Rc^{\rightarrow,\times}_{1/2+\gamma}(c)$ be the \emph{minimum} number of bits that Alice communicates to Bob, so that he can output~$c(x)$ with probability at least $1/2+\gamma$ (where the probability is taken over the randomness of Bob as well as the distribution $\mu$). Let $\Rc^{\rightarrow,\times}_{1/2+\gamma}(\Cc)=\max_{c\in \Cc} \{ \Rc^{\rightarrow,\times}_{1/2+\gamma}(c)\}$.

     We show that quantum statistical  query complexity is an upper bound on $\Rc^{\rightarrow,\mathsf{\times}}_{1/2+\gamma}(\Cc)$.       The proof follows simply by observing that Alice can simulate the quantum statistical queries and sends the outputs to Bob, who runs the learning algorithm and obtains a hypothesis $h$ that $(1/2+\gamma)$-correlates with the unknown $c\in \Cc$. Bob then  outputs $h(x)$. 
     
     	\begin{lemma}
     	\label{lem:ub-cc-qsq}
		Let $\Cc\subseteq \{c:\01^n\rightarrow \01\}$ and $\gamma>0$. Then $\Rc^{\rightarrow,\mathsf{\times}}_{1/2+\gamma}(\Cc)\leq \Qsq_{1/2+\gamma}(\Cc)\cdot \log (1/\tau)$, where the $\Qsq$ learner for $\Cc$ makes queries with tolerance at least $\tau>0$.
	\end{lemma}
	\begin{proof}
Let $\Qsq_{1/2+\gamma}(\Cc)=d$. For every $c\in \Cc$ and distribution $D:\01^n\rightarrow [0,1]$, there exists $\{(M_i,\tau)\}_{i\in [d]}$ and a $\Qsq$ learning algorithm $\A$ such that: given  $\alpha_1,\ldots,\alpha_d$ satisfying
	\begin{align}
	\label{eq:defnofalphainCC}
		\Big\vert \alpha_i-\langle \psi_c \vert M_i \vert \psi_c \rangle  \Big\vert \leq  \tau \qquad \text{ for every } i\in [d],
	\end{align}
	where $\ket{\psi_c}=\sum_x\sqrt{D(x)}\ket{x,c(x)}$, $\A$ can  output a hypothesis $h:\01^n\rightarrow \01$ satisfying 
	$\Pr_{x\sim D} [h(x)= c(x)]\geq 1/2+\gamma$.

    Consider the product distribution $\mu=\mu_1\times \mu_2$ where $\mu_1:\Cc\rightarrow [0,1]$	 and $\mu_2: \01^n\rightarrow [0,1]$. Suppose Alice receives $c\in \Cc$ according the distribution $\mu_1$ and Bob obtains $x\in \01^n$ from distribution $\mu_2$. Bob now runs the quantum statistical query protocol for the distribution $D=\mu_2$. In order to run the protocol, Alice sends $\alpha_1,\ldots,\alpha_d$ to Bob where $\alpha$s are defined in Eq.~\eqref{eq:defnofalphainCC} for the state $\ket{\psi_c}=\sum_{x}\sqrt{\mu_2(x)}\ket{x,c(x)}$ (note that the distribution $\mu$ is known both to Alice and Bob explicitly). The total cost of sending $\alpha_i$ is at most $\log(1/\tau)$.  Bob uses these $\alpha$s and obtains a hypothesis $h$ that satisfies $\Pr_{x\sim \mu_2}[h(x)= c(x)]\geq 1/2+\gamma$. Bob then outputs $h(x)$. By the promise of the $\Qsq$ algorithm, for every $c\in \Cc$, we have
$$
 \Pr_{x\sim \mu_2}[h(x)=c(x)]\geq \frac{1}{2}+\gamma.
$$
 In particular, this implies
$
\Pr_{(c,x)\sim \mu} [h(x)=c(x)]\geq \frac{1}{2}+\gamma.
$
Hence,  we have $\Rc^{\rightarrow,\times}_{1/2+\gamma}(\Cc)\leq d\log (1/\tau)$, thereby proving the lemma statement.
	\end{proof}

	Similar to $\Rc^{\rightarrow,\times}_{1/2+\gamma}(\Cc)$, we can define $\Qc^{\rightarrow,\times}_{1/2+\gamma}(\Cc)$ as the \emph{quantum communication complexity} of computing $\Cc$ under product distributions, wherein Alice is allowed to send \emph{quantum bits} to Bob. We observe that $\WSQDIM(\Cc)$ can be used to lower bound  $\Qc^{\rightarrow,\times}_{1/2+\gamma}(\Cc)$.
\begin{lemma}
	\label{lem:weakQccandSQ}
		Let $\Cc\subseteq \{c:\01^n\rightarrow \01\}$. For every $\gamma>0$, we have
		\begin{align}
		    \label{eq:our-lb-cc}
		\Omega(\log (\gamma\cdot \sqrt{\WSQDIM(\Cc)}))\leq \Qc^{\rightarrow,\times}_{1/2+\gamma}(\Cc).		    
		\end{align}
	\end{lemma}
	\begin{proof}
The main technical tool in the proof is a combinatorial quantity
called \emph{discrepancy}. which we do not define here, but we use its connections to $\WSQDIM$ and communication complexity.
 \citewithname{Sherstov}{\cite{sherstov:halfspacematrices}} showed that for any $\Cc$ and $F_{\Cc} : \Cc \times \01^n \rightarrow \01$ such that $F(c,x) = c(x)$, we have~that
$$
\sqrt{\frac{1}{2}\WSQDIM(\Cc_F)}\leq \frac{1}{\disc^{\times}(F)}\leq 8\WSQDIM(\Cc_F)^2.
$$
Our lemma statement follows from the result of \cite{klauck:CC}, who showed that 
for any function $F:X\times Y\rightarrow [0,1]$, distribution $\mu$ and $\gamma>0$, we have
$\Qc^{\mu}_{1/2+\gamma}(F)\geq\log_2 \Big(\frac{\gamma}{\disc_{\mu}(F)}\Big)$.
	\end{proof}

Prior to our work, \citewithname{Kremer {\em et al.}}{\cite{kremer:CC}} and \citewithname{Ambainis {\em et al}}{\cite{ambainis:superdense}} related the $\VC$ dimension and communication complexity by showing that for every concept class $\Cc$\, we have
\begin{align} 
\label{eq:kremer-ambainis}
(1-\textsf{H}_2(\varepsilon))\cdot \VC(\Cc)\leq \Rc^{\rightarrow,\times}_{\varepsilon}(\Cc)\leq \VC(\Cc),
\end{align}
where $H_2(\cdot)$ is the binary entropy function. In particular for constant $\varepsilon$, they showed the characterization $ \Rc^{\rightarrow,\times}_{\varepsilon}(\Cc)=\Theta(\VC(\Cc))$. 
 A priori it might seem that the lower bound of $\Omega(\log \WSQDIM(\Cc))$ in Lemma~\ref{lem:weakQccandSQ} is exponentially worse than the upper bound that we get in Lemma~\ref{lem:ub-cc-qsq} and might not be useful in comparison to the lower bound in  Eq.~\eqref{eq:kremer-ambainis}. However note that for every $\Cc$, the best lower bound that Eq.~\eqref{eq:kremer-ambainis} can yield is  $\VC(\Cc)\leq \log|\Cc|$, and one can obtain a similar lower bound using our Lemma~\ref{lem:weakQccandSQ} since $\WSQDIM(\Cc)$ could be as large as $|\Cc|$. In fact we show that in the small-error regime, our lower bound {can be} \emph{exponentially} better than what can get from Eq.~\eqref{eq:kremer-ambainis}.
 
 Suppose $\varepsilon=1/2+\gamma$ for some $\gamma\ll 1/2$ in Eq.~\eqref{eq:kremer-ambainis}. The lower bound scales then as
    $$
    (1-\textsf{H}_2(\varepsilon))\cdot \VC(\Cc)=\Big(1-\textsf{H}_2\big(\frac{1}{2}+\gamma\big)\Big)\cdot \VC(\Cc)=\Theta(\gamma^2\cdot \VC(\Cc)),
    $$
    where we used the Taylor series expansion of $\textsf{H}_2(\cdot)$ to conclude $\textsf{H}_2(1/2+\gamma)=\Theta(\gamma^2)$ for  $\gamma\ll~1/2$. Let $\Cc=\PARITY_n$ and  $\gamma=\WSQDIM(\Cc)^{-1/3} = 2^{-n/3}$, then Eq.~\eqref{eq:kremer-ambainis} gives us the~trivial
$$
\Rc^{\rightarrow,\times}_{\frac{1}{2}+\gamma}(\Cc)\geq \frac{\VC(\Cc)}{\WSQDIM(\Cc)^3}=\Omega(n\cdot 2^{-2n/3}),
$$ 
however Eq.~\eqref{eq:our-lb-cc} gives us a stronger bound of  $\Qc^{\rightarrow,\times}_{\frac{1}{2}+\gamma}(\Cc)\geq \Omega\left(\log\left(\WSQDIM(\Cc)^{\frac{1}{6}}\right)\right) = \Omega(n)$. Notice that this allows us to give non-trivial lower bounds on the communication complexity even for {\em inverse exponential bias}.
	\section{Quantum learning in a differential private setting}\label{sec:introdp}

In this section we describe the connections between the $\Qsq$ model and private learning. We start with a brief overview of classical differential privacy.

\subsection{Differential privacy}
\label{sec:dp}
Differential privacy is an important framework that provides a mathematical model for the  notion of privacy of individuals on database queries~\cite{dwork:intro,dwork:intro1,dwork:intro3,blum:intro2}. More concretely,  an algorithm $\A$ is said to be $\alpha$-differentially private if for any two neighbor databases\footnote{We can see a database $X$ as a string  in $\Sigma^n$, for some alphabet $\Sigma$.}  $X$ and $X'$, where two databases are neighbors if they differ in a single position, and for every subset $\mathcal{F}$ of the possible outcomes of $\A$ we have  
$$
\Pr\big[\A(X)\in \mathcal{F} \big]\leq e^{\alpha}\Pr\big[\A(X')\in \mathcal{F} \big].
$$

Given the success of differential privacy (in theory and practice), this notion was extended also to learning algorithms by Kasiviswanathan {\em et al.}~\cite{klnrs}. In this setting, we extend the requirements of standard PAC learning to require the learning algorithm to be differentially-private. Classically, it is well known that that if a concept class can be learned in statistical query model, it can be private PAC learnable and this connection has provided many consequences {(which we do not discuss here,  and refer the interested reader to~\cite{klnrs,beimel:deterministic,vadhan:survey} for more on differential privacy and its applications)}.

	\subsubsection{Laplacian mechanism}\label{sec:laplace}
The Laplacian mechanism is a technique used often to ensure that the output of a classical algorithm is differentially-private. The mechanism works as follows: suppose we want to compute function $f:[0,1]^T\rightarrow [0,1]$ whose input variables have small influences, then the Laplacian mechanism first computes $f$ on an  input $(x_1,\ldots,x_T)$, then adds noise from the Laplace distribution to $f(x_1,\ldots,x_T)$ and outputs the resulting value. 
	
	\begin{definition}[Laplacian mechanism]
	\label{def:laplacenoise}
	Let $T\geq 1$, $f:[0,1]^T\rightarrow [0,1]$ and $a_1,\ldots,a_T \in [0,1]$. The \emph{Laplacian mechanism} for computing $f$, first computes ${a'}=f(a_1,\ldots,a_T)$
	and outputs ${a'}+x$ where $x$ is drawn from the Laplacian distribution $D:\R\rightarrow [0,1]$ with parameter $\alpha n$ defined as
	$$
	D_{\alpha\cdot n}(x)=\frac{\alpha n}{2}e^{-|x| \cdot \alpha n},
	$$
	where $|x|$ is the absolute value of $x$.
	\end{definition}
	
	The output of the Laplacian mechanism can be shown to compute $f$ in a differentially private~manner. In particular,  
	it is well-known that it can be used to compute the average of numbers (i.e., given $a_1,\ldots,a_T\in \R$, compute ${a'}  =f(a_1,\ldots,a_T)=\frac{1}{T}\sum_{i=1}^T a_i$) privately.
	
	\begin{lemma}
	\label{lem:avgisDP}
	The Laplacian mechanism for computing the average of $T$ numbers  $\{a_1,\ldots,a_T\}$ with Laplacian parameter $\alpha \cdot T$ is $\alpha$-differentially private. Moreover, there exists some universal constant $C>0$ such that with probability at least $1-\delta$ the output $v$ of the Laplacian mechanism satisfies:
	$$
	\Big|v-\frac{1}{T}\sum_{i=1}^T a_i  \Big|\leq C\cdot  \frac{1}{\alpha T} \log \Big(\frac{1}{\delta}\Big).
	$$
	\end{lemma}
For a proof of this lemma and additionally applications of the Laplacian mechanism in differential privacy, we refer the reader to~\cite[Section~3.3]{dwork:book}.

\subsection{Private quantum PAC learning}
 Given the success of differential privacy, its  quantum analogue was recently proposed by \citewithname{Aaronson and Rothblum}{\cite{aaronson:gentle}}, which we define now.

\vspace{-0.2em}
    \begin{definition}
	Two product states  $\ket{\Phi}=\ket{\phi_1}\otimes \cdots \otimes \ket{\phi_n}$. 
	and $\ket{\Psi} = \ket{\psi_1}\otimes \cdots \otimes\ket{\psi_n}$ are neighbors \emph{if} there exists at most one $i\in [n]$ such that $\ket{\phi_i}\neq \ket{\psi_i}$.	A quantum algorithm $\A$ is $\alpha$-differential private on some subset $S$ of product states if for all states $\ket{\Phi},\ket{\Psi} \in S$ that are neighbors and every subset $\mathcal{F}$ of the possible outputs of $\A$ we have that\footnote{Following \cite{beimel:randomized}, we define the stronger notion of privacy where the probabilities are close not only for every possible output, but also for every {\em subset} of outputs.} 
	\[\Pr[\A(\ket{\Psi}) \in \mathcal{F} ] \leq e^{\alpha} \Pr[\A(\ket{\Phi}) \in \mathcal{F} ]. \]
	\end{definition}
	Inspired by this definition, we now define private learning a concept class.
	
\vspace{-0.3em}	
	\begin{definition}
		Let $\Cc$ be a concept class. We say a $\A$ is a \emph{ $(\alpha,\varepsilon,\delta)$-differentially private  quantum PAC} learning algorithm for $\Cc$ with sample complexity $T$ if $i)$ $\mathcal{A}$ is $\alpha$-differentially private and $ii)$
for every distribution $D:\01^n\rightarrow [0,1]$, $\A$ uses $T$ copies of $\ket{\psi_c}=\sum_x\sqrt{D(x)}\ket{x,c(x)}$, and with probability at least $1-\delta$ outputs $h$ such that 
			$
			\Pr_{x\sim D} [h(x)\neq c(x)]\leq \varepsilon.
			$
	\end{definition}	

In classical literature it is well-known that if a concept class is learnable in the $\Sq$ model, then it can also be learned privately in the PAC learning model. We now show that this implication also holds true in the quantum case. The proof follows similarly to \Cref{claim:kearnssqimpliesnoisy}: we use
	$O(\tau^{-2})$ quantum examples to simulate a $\Qstat$ oracle with tolerance $\tau$ and then, to ensure privacy the of each query, we use the well-known Laplacian mechanism (see \Cref{sec:laplace}) in the simulation of $\Qstat$ by quantum examples.
\begin{theorem}\label{thm:qsq-implies-qdp}
    Let $\Cc \subseteq \{c:\01^n\rightarrow \01\}$. If there exists a learning algorithm that $\eps$-learns $\Cc$ using $d$ $\Qstat$ queries with tolerance at least $\tau$, then the quantum sample complexity of $(\alpha,\varepsilon,\delta)$-private quantum PAC learning $\Cc$ is $O\Big( \Big(\frac{d}{\tau^2}+\frac{d}{\varepsilon \tau}\Big)\cdot \log\Big(\frac{d}{\beta}\Big)\Big)$.
\end{theorem}
\begin{proof}
 The proof here is similar to the proof of Theorem~\ref{claim:kearnssqimpliesnoisy} where we showed quantum $\Sq$ learnability implies quantum PAC learnability. Suppose $\Qsq(\Cc)=d$. For every $c\in \Cc$ and distribution $D:\01^n\rightarrow [0,1]$: suppose the $\Qsq$ learner $\A$  makes the queries $\{(M_i,\tau)\}_{i\in [d]}$ and obtains  $\alpha_1,\ldots,\alpha_d$ satisfying
	$$
	\Big\vert \alpha_i-\langle \psi_c \vert M_i \vert \psi_c \rangle  \Big\vert \leq  \tau \qquad \text{ for every } i\in [d],
	$$
	for $\ket{\psi_c}=\sum_x\sqrt{D(x)}\ket{x,c(x)}$, then $\A$ outputs a hypothesis $h$ satisfying 
	$\Pr_{x\sim D} [h(x)\neq c(x)]\leq \eta$.
Let 
$$
Q=C \alpha^{-1} \cdot \Big(\frac{1}{\tau^2}+\frac{2}{\tau}\Big)\cdot \log \Big(\frac{2d}{\delta}\Big)
$$
where $C$ is the constant defined in \Cref{lem:avgisDP}.
Consider a quantum PAC learner that for every $i\in~[d]$, the  learner  obtains $Q$
many (fresh) quantum examples $\ket{\psi_c}$, which are measured according to the observable $M_i$ with outcomes $a^i_1,\ldots,a^i_Q$. The learner then applies the Laplacian mechanism $\mathsf{LM}$ to compute the average of $a^i_1,\ldots,a^i_{Q}$: first compute ${b^i}=\sum_{j=1}^Q a^i_j$ and then apply to ${b^i}$ the Laplacian noise   with parameter $\alpha\cdot Q$, resulting in $\tilde{b}^i$ (see Definition~\ref{def:laplacenoise}). The quantum PAC learner feeds the $\Qsq$ learner with $\{\widetilde{b}^1,\ldots,\widetilde{b}^d\}$, and outputs the hypothesis $h$ provided by the $\Qsq$ learner. The sample complexity of this PAC learner is $O\left(d \alpha^{-1} \cdot \Big(\frac{1}{\tau^2}+\frac{2}{ \tau}\Big)\cdot \log\Big(\frac{d}{\delta}\Big)\right)$.

We first analyze the correctness of our quantum PAC learner. Similar to the proof of Theorem~\ref{claim:kearnssqimpliesnoisy}, observe that $Q$ is large enough to ensure that, with probability at least $1-\delta/(2d)$, we have $|b^i-\langle \psi_c \vert M_i \vert \psi_c \rangle|\leq \tau/2$ for every $i$. Next, by Lemma~\ref{lem:avgisDP}, with probability at least $1-\delta/(2d)$, we have
$$
\Big|\tilde{b}^i-{b^i}\Big|\leq C \cdot \frac{1}{\alpha Q} \cdot \log \Big(\frac{2d}{\delta}\Big) \leq \frac{\tau}{2},
$$
where the last inequality used the definition of $Q$. The difference between the quantum $\Sq$ query response and $\tilde{b}^i$ can be bounded using the triangle inequality by 
$$
|\langle \psi_c \vert M_i \vert \psi_c \rangle-\tilde{b}^i|\leq |\langle \psi_c \vert M_i \vert \psi_c \rangle-b^i|+|b^i-\tilde{b}^i|\leq \tau.
$$
Moreover, by a union bound we have: with probability at least $1-\delta$, the quantum PAC learner answers all $d$  $\Qsq$ queries with error at most $\tau$. Hence, with probability $\geq 1-\delta$, the output~$h$ of the quantum $\Qsq$ learner (and hence the quantum PAC learner) satisfies $\Pr_{x\sim D} [h(x)\neq c(x)]\leq \eps$.

    We now analyze the privacy of our quantum PAC learner. For that, let us analyze the privacy for computing $\tilde{b}^i$ for some fixed $i$. Let $\mathcal{Q}$ be the procedure that computes $\tilde{b}^i$ from $\ket{\psi_c}^{\otimes Q}$. It follows~that
    \begin{align*}
    \Pr [\mathcal{Q} \big(\ket{\psi_c}^{\otimes Q}\big) = y]&=\Pr_{a_1^i,\ldots,a_Q^i} [\mathsf{LM } \big(\{a_1^i,\ldots,a_Q^i\}\big) = y]\\
    &\leq e^{\alpha}\cdot \Pr_{a_1^i,\ldots,a_{Q}^i}  [\mathsf{LM } \big(\{a_1^i,\ldots,a_{Q-1}^i,w\}\big) = y]= e^{\alpha} \cdot \Pr [\mathcal{Q} \big(\ket{\psi_c}^{\otimes Q-1}\otimes \ket{\phi}\big) = y],
\end{align*} 
where we use that the Laplacian mechanism with our parameters is $\alpha$-differential private and we assume for simplicity that the (possibly) different entry is the last state in the tensor product.    

Notice that the quantum PAC learning algorithm $\A$ receives as input  $\ket{\psi_c}^{\otimes Qd}$, and runs $\mathcal{Q}$ $d$-times in parallel and runs a procedure $\mathcal{S}$ that computes the hypothesis basis on the classical statistics.
Let assume again for simplicity that the neighbor $\ket{\Phi}$ of $\ket{\psi_c}^{\otimes Qd}$ has its different entry in the last position. 
In this case, for any subset of outputs $\mathcal{F}$, we have that 
    \begin{align*}
 &\Pr [\mathcal{A} \big(\ket{\psi_c}^{\otimes Qd}\big) \in \mathcal{F}]\\
    &=  \Pr[\mathcal{S}(y_1,...,y_d)  \in \mathcal{F}]
    \Pr [\mathcal{Q} \big(\ket{\psi_c}^{\otimes Q} \big) = y_1] \cdots
        \Pr [\mathcal{Q} \big(\ket{\psi_c}^{\otimes Q} \big) = y_d] \\
  &\leq e^{\alpha}
\Pr[\mathcal{S}(y_1,...,y_d)  \in \mathcal{F}]
    \Pr [\mathcal{Q} \big(\ket{\psi_c}^{\otimes Q} \big) = y_1] \cdots
        \Pr [\mathcal{Q} \big(\ket{\psi_c}^{\otimes Q-1}\otimes \ket{\phi} \big) = y_d] \\
  &=     \Pr [\mathcal{A} \big(\ket{\Phi} \big) \in \mathcal{F}],
\end{align*} 
showing that $\A$ is also $\alpha$-private. 
\end{proof}

An immediate corollary of this theorem along with the results in Section~\ref{sec:introlearneff} is the following  (which was not known before).
\begin{corollary}
 Parities, $k$-juntas and DNFs can be \emph{privately} quantum PAC learned under the uniform~distribution.
\end{corollary}

\subsection{Representation dimensions and private quantum PAC learning}

It is well-known that the sample complexity of classical and quantum PAC learning is characterized by $\VC$ dimension~\cite{blumer:vc,hanneke:pac,arunachalam:optimalpaclearning}. Classically, in the private setting, a series of results \cite{beimel:randomized,beimel:deterministic,feldmanxiao:dpcc}  showed that the \emph{representational dimension} of the concept class~$\Cc$ ($\Prdim(\Cc)$) characterizes the sample complexity of \emph{private} PAC learning. Here, we show that $\Prdim(\Cc)$ also characterizes the sample complexity of private {\em quantum} PAC learning $\Cc$.

In order to define the representation dimension of a concept class $\Cc\subseteq \{c:\01^n\rightarrow \01\}$, we first define the \emph{probabilistic representation} of~$\Cc$  and its \emph{probabilistic representational~dimension}.
	
	\begin{definition}[Representation of concept classes]
	A hypothesis class $\Hi\subseteq \{h:\01^n\rightarrow \01\}$ is an \emph{$\varepsilon$-representation of $\Cc$} if for every $c\in \Cc$ and distribution $D:\01^n\rightarrow [0,1]$, there exists $h\in \Hi$ such that $\Pr_{x\sim D}[h(x)\neq c(x)]\leq \varepsilon$. 
		
		Similarly, let $P:[r]\rightarrow [0,1]$ be a distribution and  $\Ch=\{\Hi_1,\ldots,\Hi_r\}$ be a collection of hypothesis classes. We say $(P,\Ch)$ is an \emph{$(\varepsilon,\delta)$-probabilistic representation} of $\Cc$, if for every $c\in \Cc$ and distribution $D:\01^n\rightarrow [0,1]$, we have
		$$
		\Pr_{i\sim P} [\exists h\in \Hi_i \text{ s.t. } \Pr_{x\sim D}[h(x)\neq c(x)]\leq \varepsilon ]\geq 1-\delta.
		$$	
		Define $\size(\Ch)=\max\{\log |\Hi_i|: \Hi_i\in \Ch\}$
	\end{definition}

We are now ready to define the probabilistic representational dimension of a concept class.	
	
	\begin{definition}[Representational dimension~\cite{beimel:deterministic,beimel:randomized}]
		Let $\Cc\subseteq \{c:\01^n\rightarrow \01\}$ be a concept class.  The $(\varepsilon,\delta)$ probabilistic representational dimension of $\Cc$, $\Prdim(\Cc)$ is defined as
				$$
		\min \Big\{\size(\Ch): \text{ there exists } (P,\Ch)  \text{ that } (\varepsilon,\delta)-\text{probabilistically represents } \Cc \Big\},
		$$

		\end{definition}

We now show that for every concept class $\Cc$, the quantum sample complexity of \emph{private} PAC learning $\Cc$ is characterized by the representation dimension of a concept class.  Since $\Prdim(\Cc)$ is an upper-bound to the {\em classical} sample complexity of private PAC learning (which in its turn is an upper bound to the quantum sample complexity\footnote{In particular, this inequality holds because the following algorithm is a {\em private quantum learner}: suppose a quantum learner obtains $T$ quantum examples, measures each quantum example in the computational basis and then runs the classical private learning algorithm on the $T$ classical examples. This quantum algorithm satisfies the conditions of quantum differential privacy, because a neighboring quantum state that is provided to the quantum learner will result in neighboring classical examples and by assumption we know that the classical learner is differentially private.}), we only need to show that $\Prdim(\Cc)$ is also a {\em lower bound} on the quantum sample complexity of quantum private PAC learning.  Together with the corresponding classical characterization~\cite{beimel:deterministic,beimel:randomized} (which inspires our proof), our result implies that quantum and classical sample complexities of private PAC learning are equal, up to constant~factors.

	\begin{theorem}
 If there exists an $(\alpha,\varepsilon,\delta)$-quantum private PAC learner  for a concept class $\Cc$ with sample complexity $T$, then the $(\varepsilon,\beta)$-probabilistic dimension $\Prdim(\Cc)=O(T\alpha++\log \log 1/\beta)$.\footnote{One can further prune this bound to get the $\varepsilon$ dependence in the upper bound on $\Prdim(\Cc)$ by using ideas in~\cite[Lemma~3.16]{beimel:randomized}, we omit it here.}
	\end{theorem}
	\begin{proof}
		Let $\A$ be a $(\alpha,\varepsilon,1/2)$-quantum private learning algorithm for $\Cc$ using a hypothesis class~$\Fe$, with sample complexity $T$. Fix $c\in \Cc$ and distribution $D$ and define $\Fe'\subseteq \Fe$  as $\Fe'=\{h\in \Fe: \Pr_{x\sim D}[c(x)\neq h(x)]\leq \varepsilon\}$.   By the ``$\delta$-learning promise" of $\A$, we know 
		\begin{align}
		    \label{probofinFe}
		\Pr \big[\A\big(\ket{\psi_c}^{\otimes T}\big)\in \Fe'\big]\geq 1-\delta,
				\end{align}
		where the probability is taken with respect to the randomness of $\A$.
		Let  $\ket{\psi_{\textbf{0}}}=\sum_{x}\sqrt{D(x)}\ket{x,0}$. The  $\alpha$-quantum differential privacy of $\A$ implies that 
	\begin{align*}
	\Pr[\A\big(\ket{\psi_{\textbf{0}}}^{\otimes T}\big)\in \Fe'] &\geq  e^{-\alpha} \cdot 	\Pr[\A\big(\ket{\psi_{\textbf{0}}}^{\otimes T-1} \otimes \ket{\psi_{c}}\big)\in \Fe']\\
	&\geq  e^{-2\alpha} 	\Pr[\A\big(\ket{\psi_{\textbf{0}}}^{\otimes T-2}\otimes \ket{\psi_{c}}^{\otimes 2}\big)\in \Fe']\geq \cdots  \geq  e^{-T\alpha} \cdot 	\Pr[\A\big(\ket{\psi_{c}}^{\otimes T}\big)\in \Fe']
	\end{align*}	
	which is at least $(1-\delta)e^{-T\alpha}$ using Eq.~\eqref{probofinFe}.  In particular, we have that $\Pr_{} \big[\A\big(\ket{\psi_{\textbf{0}}}^{\otimes T}\big)\notin \Fe'\big]\leq 1-(1-\delta)e^{-T\alpha}$. Suppose, we run $\A$ $K=\ln(1/\beta)\cdot e^{T\alpha}/(1-\delta)$ many times on input $\ket{\psi_{\textbf{0}}}^{\otimes T}$, and let~$\Hi$ be the set of the outcomes of $\A$ on each execution. The probability that $\Hi$ does not contain an $\varepsilon$-good hypothesis is at most 
		$$
		\Big(1-(1-\delta)\cdot e^{-T\alpha}\Big)^{K}\leq \exp(-K(1-\delta)e^{-T\alpha})\leq  \beta,
		$$
		 using $(1-x)^t\leq e^{-xt}$ in the first inequality and the definition of $K$ in the second inequality. Let $\widetilde{\Fe}\subseteq \Fe$ be the set of hypothesis that have a non-zero probability of being output when $\A$ is given the input~$\ket{\psi_{\textbf{0}}}^{\otimes T}$. Let also $\Ch=\big\{\Hi\subseteq \widetilde{\Fe}: |\Hi|\leq \ln(1/\beta)\cdot e^{T\alpha}/(1-\delta)\big\}$ and $P$ be the uniform distribution over all $\Hi\in \Ch$. Then $(P,\Ch)$ is an $(\varepsilon, \beta)$-probabilistic representation for the concept class $\Cc$ and it follows that 
		$$
		\Prdim(\Cc)\leq  \max_{\Hi\in \supp(\Ch)} \{\ln |\Hi|\}\leq O(T \alpha+\log \log 1/\beta),
		$$
		which proves the theorem statement.
	\end{proof}

	\section{Discussion and future work}


{An important open question is, does access to many copies of quantum examples  increase the power of the quantum learning? More concretely, we can rephrase this question as, can separate quantum PAC learning and quantum statistical query learning (even under the uniform distribution).} The classical analogue of this question can be answered using the concept class of parity functions (which can be PAC learned in classical polynomial time and requires exponential time in the SQ framework). However, quantumly, as far as we are aware, all concept classes that are learnable in the quantum PAC setting seem to be learnable in the quantum SQ setting. A positive answer to this question, would help shed light on the question if entanglement is necessary for  quantum learning{, as recently shown for quantum property testing~\cite{BCL20}.}

In this paper we considering the learnability of Boolean functions using quantum statistical queries. But one could also consider the leranability of \emph{quantum states} in the $\Qsq$ model: let $\Cc$ be a class of $n$-qubit quantum states $\rho$; a $\Qsq$ algorithm for $\Cc$ can specify a two-outcome POVM~$M$ and obtains an addictive approximation to $\Tr(M\rho)$ (for the unknown target state $\rho \in \Cc$). Using such statistical queries, can we learn $\Cc$ in the PAC setting? We remark that many algorithms for learning quantum states~\cite{aaronson:gentle,aaronson:qlearnability,anshu2020sample,chunglin2018:channel,aaronson2018online,rocchetto:stab} can be phrased in terms of the statistical query~model. It would be interesting to see if various results present in our paper also carry over to the setting of learning quantum states. 

Classically, it is well-known that \emph{many} algorithms used in practice can be implemented using simply a statistical oracle, for example expectation maximization, simulated annealing, gradient descent, support vector machine, markov chain monte carlo methods, principle component analysis, convex optimization (see~\cite{reyzin2020statistical,feldman2017statistical} for more applications and references regarding these connections). It would be interesting if we could also phrase the quantum algorithms for these problems as well in the \emph{quantum} statistical query framework.  If so, $\Qsq$ learning would provide a unified framework for understanding theoretical and practice quantum algorithms in machine learning.

We now raise two further open questions in our $\Qsq$ framework: In our definition of $\Qsq$, we have a  classical randomized learner and one could possibly consider a general definition of $\Qsq$ model wherein the algorithm can make \emph{quantum superposition queries} to the oracle, or ask the oracle to perform joint, entangling measurements on multiple copies of $\ket{\psi_{c^*}}$. Secondly, Bun and Zhandry~\cite{BunZ16} showed that classical PAC learning is strictly more powerful than its private version under cryptographic assumptions. We leave understanding if such a separation also works in a (post-) quantum scenario as an open~question.
\newcommand{\etalchar}[1]{$^{#1}$}

\end{document}